\documentclass[12pt, twoside]{article}
\usepackage{amsmath}
\usepackage{amsthm}
\usepackage{setspace}
\usepackage{verbatim}
\usepackage{lmodern,microtype}
\usepackage{booktabs}
\usepackage{colortbl}
\definecolor{lightgray}{rgb}{0.2,0.2,0.2}
\usepackage{caption}
\usepackage{footmisc}
\usepackage{lipsum}

\usepackage{xcolor}

\usepackage[left=1.4in,right=1.4in,top=1.4in,bottom=1.4in]{geometry}
\usepackage{soul}

\usepackage{xfrac}

\usepackage{amsfonts}

\usepackage{qtree}

\usepackage{tikz}

\usepackage{multirow,array}

\usepackage[unicode=true,pdfusetitle,
bookmarks=true,bookmarksnumbered=false,bookmarksopen=false,
breaklinks=false,pdfborder={0 0 1},backref=true,colorlinks=true,citecolor=blue,linkcolor=blue]
{hyperref}

\usepackage{booktabs} 
\usepackage{subcaption}

\usepackage{latexsym}
\usepackage{float}
\usepackage{graphicx}
\usepackage{enumitem}
\usepackage[font=small,labelsep=none]{caption}

\usepackage[round]{natbib}
\bibpunct{(}{)}{;}{a}{}{}
\setcitestyle{notesep={: }}

\usepackage[title]{appendix}
\usepackage{graphicx}

\usepackage{CJK}

\newtheorem*{proposition*}{Proposition}

\newtheorem{proposition}{Proposition}

\usepackage{todonotes}

\newcommand{\lotOneQuant}{L_1}

\newcommand{\lotOneEx}{L_1}

\theoremstyle{definition}

\theoremstyle{definition}

\definecolor{lightgray}{gray}{0.9}

\sethlcolor{lightgray}

\newif\ifanon
\anonfalse

\ifanon
\title{A Dominance Argument Against Incompleteness}
\date{}
\author{}

\else

\title{A Dominance Argument Against Incompleteness}
\date{April 2025}
\author{Christian Tarsney, Harvey Lederman, and Dean Spears}

\fi

\begin{document}
	
	\maketitle

\begin{center}
    Forthcoming in \textit{Philosophical Review}.

    \vspace{4mm}
\end{center}
    
	\begin{abstract}
		This article presents a new argument against many forms of moral and prudential value incompleteness. The argument relies on two central principles: (i) a weak ``negative dominance'' principle, to the effect that Lottery 1 is better than Lottery 2 only if some possible outcome of Lottery 1 is better than some possible outcome of Lottery 2, and (ii) a weak form of \textit{ex ante} Pareto, to the effect that, if Lottery 1 gives an unambiguously better \textcolor{black}{(stochastically dominant)} prospect to some individuals than Lottery 2, and equally good  prospects to everyone else, then Lottery 1 is better than Lottery 2. Given modest auxiliary assumptions, these two principles rule out incompleteness in the prudential ranking of individual lives, and many forms of \textcolor{black}{incompleteness} in the moral rankings of outcomes and lotteries. 
	\end{abstract}

	\textit{Keywords} value theory, population ethics, incompleteness, incomparability, dominance, Pareto, opaque sweetening
	
	\doublespacing 
	
	\section{Introduction}
	\label{s:intro}
	
	It often seems, when comparing two alternatives,   
	that neither is better than the other, nor are they equally good.  
	Suppose that Patricia is deciding whether to pursue a career as an artist or a banker. Being an artist would give her an unlimited creative outlet and the freedom of being her own employer. Being a banker would give her financial stability and regular opportunities for international travel. Even setting aside her many uncertainties and imagining maximally specific ways that each of these lives might go, it seems that neither life would be \emph{better} for Patricia overall---and that they would not be equally good, since a slight but noticeable improvement to one life (say, adding \$1,000 in salary) would not be enough to make it better than the other. Similarly, from an impersonal point of view (``the point of view of the universe'', or of social welfare), the world where Patricia becomes an artist could be neither better nor worse nor equally as good as the world where she becomes a banker.\footnote{Philosophers who have argued that lives can be incomparable in this way include \citeauthor{raz1985value} (\citeyear{raz1986morality,raz1985value}), \citet{chang2002possibility}, and \citet{hedden2023dimensions}, among many others.}
	
	Many philosophers have claimed that a structurally similar phenomenon arises when comparing outcomes in which different individuals exist. Some have claimed that, from a personal point of view, existence is always incomparable with non-existence: no matter how good or bad an individual's life, we can never say that it is better, worse, or equally as good for her as non-existence \citep[168]{broome1999ethics}.  
	Some have claimed that the same is true from an impersonal point of view: two possible worlds which differ only by the existence of one individual are incomparable in terms of overall value, whatever life that individual lives in the world where they exist (\citeauthor{bader2022person} \citeyear{bader2022person,bader2022asymmetry,baderbook}). Others have claimed that this is true only of \emph{some} lives: \textcolor{black}{two worlds which differ only by the existence of one individual are incomparable in terms of overall value if that individual's life contains the right mix of triumphs and tragedies in the world where they exist} \citep{gustafsson2020population,thornley2022critical}.\footnote{\textcolor{black}{Although some philosophers use the term ``incomparable'' more narrowly, we will use it simply to mean that one thing is neither better than, nor worse than, nor equal to another (even though they are both in the domain of the relevant value concept, e.g.\ personal or impersonal betterness). 
		}
	}
	
	This paper presents a new argument against these and other forms of incomparability.\footnote{\textcolor{black}{Our arguments will concern incomparability in the personal/prudential value of lives, and the impersonal/moral value of outcomes and lotteries. For concision, we will use ``value'' without qualification to refer to these particular forms of value. But our argument will not in any obvious} \textcolor{black}{way apply to other possible forms of value, e.g.\ aesthetic or epistemic.}} \textcolor{black}{The argument is based on the following sort of example:} Suppose that two possible lives, like \emph{artist} and \emph{banker}, are indeed incomparable. Let $a$ denote a fully specified artist's life and $b$ a fully specified banker's life, with $a^+$ a slightly improved version of the artist's life that is still incomparable with $b$. Now suppose we face a choice between the following two lotteries, based on the same fair coin flip, where our choice will affect only two individuals, $P_1$ and $P_2$:

	\begin{table}[!htb]
		
		\centering
		
		\begin{minipage}{.3\linewidth}
			\begin{subtable}{\linewidth}
				\centering
				\caption{\textsc{Lottery 1}}
				\begin{tabular}{lcc}
					\toprule
					& Heads & Tails \\
					\midrule
					$P_1$	& $a$ & $b$ \\
					$P_2$	& \textcolor{black}{$a$}   & \textcolor{black}{$b$} \\	
					\bottomrule
				\end{tabular}
			\end{subtable}
			
		\end{minipage}
		\begin{minipage}{.3\linewidth}
			\begin{subtable}{\linewidth}
				\centering
				\caption{\textsc{Lottery 2}}
				\begin{tabular}{lcc}
					\toprule
					& Heads & Tails \\
					\midrule
					$P_1$	& $a^+$ & $b$ \\
					$P_2$	& \textcolor{black}{$b$}	& \textcolor{black}{$a$} \\	
					\bottomrule
				\end{tabular}
			\end{subtable}		
		\end{minipage}
		
	\end{table}
	
	There is, on the one hand, a strong argument that neither lottery is better than the other. First, if $a$ and $b$ are incomparable from a personal point of view, then it seems that if two outcomes differ only by the fact that some individual enjoys life $a$ in the one and life $b$ in the other, those outcomes should be incomparable. Thus, $(a,a)$ and $(b,b)$ are both incomparable with $(a,b)$ and $(b,a)$, where the first and second terms in each ordered pair represent the lives of $P_1$ and $P_2$ respectively. Moreover, it is part and parcel of incomparability that it is preserved by small enough improvements or worsenings, so $(a^+,b)$, which is at most slightly better than $(a,b)$, should likewise be incomparable with $(a,a)$ and $(b,b)$.\footnote{One might think that $(a^+,b)$ is better than $(a,a)$ on the grounds that it is better for someone and worse for no one. But this reasoning, though intuitive, cannot be sound, since it would lead to cycles of betterness: for instance, given incomparable lives $a$ and $b$, we would get that $(a,a)$ is better than $(a^-,b)$, which is better than $(b,b^-)$, which is better than $(b^-,a^+)$, which is better than $(a,a)$. \textcolor{black}{(For an extended exploration of this point and its implications, see \cite{hedden2024parity}.)}\label{fn:superStrongPareto}} We then find that \emph{every} possible outcome of Lottery 1 (the Heads outcome and the Tails outcome) is incomparable with \emph{every} possible outcome of Lottery 2. This suggests that neither lottery is better than the other.
	
	On the other hand, there is \emph{also} a  strong argument that Lottery 2 is better than Lottery 1: Lottery 1 gives both $P_1$ and $P_2$ a 50\% chance of life $a$ and a 50\% chance of life $b$. Lottery 2 does the same, except that it replaces $P_1$'s 50\% chance of life $a$ with a 50\% chance of life $a^+$. Since $a^+$ is better than $a$, it seems that Lottery 2 is better from the perspective of $P_1$, equally good from the perspective of $P_2$, 
    and therefore better overall.
	
	Taken together, these arguments amount to a  
	\emph{reductio} of the supposition that $a$ and $b$ are incomparable. Similar arguments, as we will see, can be mounted against other forms of incomparability in the axiological ranking of finite social outcomes, for instance between larger and smaller populations or between competing desiderata like total welfare, average welfare, and equality.\textcolor{black}{\footnote{Although our arguments and the principles behind them could formally be extended to some infinitary contexts---in particular, to finite-support lotteries over infinite populations---we will restrict our claims to strictly finite contexts for reasons that will become clear in the second half of the paper (see especially n. \ref{infinitenote} in \S \ref{s:rejectingNegativeDominance}).}} The result will be new arguments against many forms of value incomparability, 
    turning on considerations of a quite different character than, e.g., the recent arguments of \citeauthor{dorr2021consequences} (\citeyear{dorr2021consequences,dorr2023case}). \textcolor{black}{The central premises of this \emph{reductio} are two dominance-like principles: first, that if no possible outcome of lottery $\lotOneQuant$ is better than any possible outcome of lottery $L_2$, then $\lotOneQuant$ is not better than $L_2$; and second, that if $\lotOneQuant$ is unambiguously at least as good for every individual as $L_2$, and unambiguously better for some (satisfying restricted forms of weak/strict stochastic dominance), then it is better overall.}  
	We will call the first principle \emph{Negative Dominance} and the second principle \emph{Personal Good}.
	
	In \S \ref{s:negativeDominance}, we formalize and motivate Negative Dominance, comparing it to a stronger principle which plays a central role in Caspar Hare's ``opaque sweetening'' puzzle \citep{hare2010take}. 
	In \S \ref{s:personalGood}, we similarly \textcolor{black}{formalize and motivate} Personal Good. In \S 
	\ref{s:resultsForPersonalIncompleteness}, we offer a more exact version of the above argument, showing that, against modest background assumptions, Negative Dominance and Personal Good  rule out incompleteness in the personal ranking of lives, and by extension those forms of incompleteness in the impersonal ranking of outcomes that derive from incomparability between lives. \S \ref{s:resultsForPersonalCompleteness} then presents a more general argument against incompleteness in the impersonal ranking of outcomes, likewise based on Negative Dominance and Personal Good. \S \ref{s:lotteries} takes the argument a step further by showing that Personal Good, together with plausible auxiliary assumptions, allows the ranking of a large class of lotteries to be reduced to the ranking of outcomes, so that if outcomes are completely ordered by impersonal value, this class of lotteries will be too. 
	\textcolor{black}{\S \ref{s:rejectingNegativeDominance} considers and responds to three potential objections to these arguments.} \S \ref{s:conclusion} summarizes our conclusions. 
	
	\section{Negative Dominance}
	\label{s:negativeDominance}
	
	We will be interested in the implications of certain decision-theoretic principles in an ethical context, so we begin with an abstract decision-theoretic framework which we will subsequently give an ethical interpretation. At the most abstract level, we are interested in rankings of \emph{lotteries} over \emph{outcomes}. \emph{Outcomes} can be understood as complete possible worlds, as equivalence classes of equally good worlds, or as specifications of all evaluatively significant features of a world. \textcolor{black}{A \emph{lottery} is a discrete probability distribution over a finite set of outcomes. 
		(We'll consider \textcolor{black}{lotteries with infinitely many possible outcomes} in \S \ref{s:rejectingNegativeDominance}.)}   
	We assume that lotteries can be compared in terms of \emph{impersonal value}: One lottery can be \emph{impersonally better} than another, and two lotteries can be \emph{equally impersonally good}. If $\lotOneQuant$ is either better than or equal to $L_2$, we say that it's \emph{at least as (impersonally) good as} $L_2$. If neither relation holds between two lotteries, we say that they are \emph{(impersonally) incomparable}.
	\textcolor{black}{Finally, we assume that the impersonal betterness relation on lotteries corresponds to an impersonal betterness relation on outcomes: Outcome $o_1$ is impersonally better than outcome $o_2$ if and only if the lottery that yields $o_1$ with certainty is impersonally better than the lottery that yields $o_2$ with certainty (and likewise for other evaluative relations).}
	
	Our interest here will be in the ethics of population and distribution. In this context, we assume that the evaluatively significant features of an outcome are (i) which individuals exist in that outcome and (ii) how their lives go. Thus, we will identify outcomes with specifications of only these features. In particular, we will assume a countably infinite set of possible individuals (individuals who exist in some possible world/outcome) 
	and a set of possible \emph{lives} (complete specifications of everything that matters, morally or prudentially, about an individual's life---e.g., the totality of their experiences). 
	Importantly, the set of possible ``lives'' includes non-existence (``the empty life''). Given this background, we will represent outcomes as functions from the set of possible individuals to the set of possible lives, understood as specifying which individuals exist in a given outcome and what lives they lead. \textcolor{black}{More specifically, until \S \ref{s:rejectingNegativeDominance}, we will restrict the scope of our discussion to \emph{finite-population} outcomes, meaning that we understand outcomes as always assigning nonexistence to all but finitely many possible individuals.} 
	We will sometimes refer to outcomes as \emph{populations}, when we want to emphasize the fact that they have less information in them than a possible world.
	Finally, alongside the \emph{im}personal betterness relation introduced above, we will assume a \emph{personal} betterness relation on lives: one life can be personally better than another, and two lives can be equally personally good. 
	If life $a$ is either better than or equal to life $b$, we say that it's \emph{at least as (personally) good} as $b$. If neither relation holds between two lives, we say that they are \emph{(personally) incomparable}. (From now on, we will usually omit ``personal'' and ``impersonal'' where it is clear from context which relation is meant.)
	
	A betterness relation is \emph{complete} if no two entities in its domain are incomparable: that is, for any two such entities, either one is better than the other, or they are equally good. 
	
	Given these preliminaries, we can state the principle of Negative Dominance as follows:
	\begin{description}
		\item[Negative Dominance] \textcolor{black}{If none of the possible outcomes of lottery $\lotOneQuant$ are better than any of the possible outcomes of lottery $L_2$, then $\lotOneQuant$ is not better than $L_2$.}
	\end{description}

	\textcolor{black}{This principle is \emph{prima facie} extremely plausible. It's obvious that, if every possible outcome of $\lotOneQuant$ is better than every possible outcome of $L_2$, then $\lotOneQuant$ 
		must be better than $L_2$. It seems nearly as obvious that, if \emph{no} possible outcome of $\lotOneQuant$ is better than \emph{any} possible outcome of $L_2$, then $\lotOneQuant$ is \emph{not} better than $L_2$---which is what Negative Dominance asserts. We therefore think it would be reasonable to treat Negative Dominance as a bedrock principle of decision theory, on a par with other widely accepted dominance principles.}
	
	But there is also a deeper justificatory story to be told in favor of Negative Dominance. \textcolor{black}{The central idea of this story is that the evaluative ranking of outcomes is more fundamental than, and \textcolor{black}{integral to explaining}, the ranking of lotteries.} 
		The value of outcomes is \emph{actual} value---the sort of value possessed by actual joys, achievements, etc. The value of lotteries 
	is something different in kind: A 20\% chance of experiencing a moment of pure bliss is not good \emph{in the same way} as the moment itself would be. \textcolor{black}{Moreover, the thing that ultimately matters in evaluating lotteries, and grounds any comparisons between them, is their potential for actual value---that is, the values and probabilities of their potential outcomes.} \textcolor{black}{\textcolor{black}{In other words,} while outcomes have value in and of themselves (\emph{final} value, in the sense of \cite{korsgaard1983two}), lotteries have value only for the sake of their potential outcomes (a form of nonfinal value).}\footnote{\textcolor{black}{Thanks to an anonymous reviewer for suggesting we put this point in terms of final value.} The picture we sketch is inspired  by similar remarks in \citet[268]{schoenfield2014decision}, who writes: ``It is important to realize that the \emph{ultimate} good, for consequentialists, or in cases in which consequentialist reasoning is appropriate, is an outcome that is valuable, not an outcome that is expectedly valuable. A consequentialist is only interested in \emph{expected} value insofar as it helps her obtain what is actually important: value.''}

\textcolor{black}{This distinction allows us to give a first-pass argument in favor of Negative Dominance: Nonfinal value is derivative of final value. So the comparative nonfinal value of a pair of lotteries must be explicable in terms of the comparative final value of their outcomes. In particular, if $\lotOneQuant$ is (nonfinally) better than $L_2$, then the explanation of this fact must involve some possible outcome of $\lotOneQuant$ being (finally) better than some possible outcome of $L_2$.} 

\textcolor{black}{We can go beyond this rough and ready argument if we adopt a more specific picture of the nonfinal value of lotteries. To this end, we will temporarily suppose, as we think is plausible, that the betterness ranking of lotteries can be understood as an ideal preference ordering.} In particular, we suggest, to say that one lottery is impersonally better than another just means that an agent motivated by impersonal value rationally ought to prefer it---or in other words, that she has \emph{all-things-considered requiring reason} to prefer it.\footnote{For a thorough discussion of how to analyze value relations in terms of permissible preference, see \cite{rabinowicz2008value}---though Rabinowicz applies this analysis to value relations generally, while our analysis is restricted to the value of lotteries. \textcolor{black}{While we ourselves are inclined to see the value of outcomes as a conceptual primitive, 
		our argument is compatible with analyzing the value of outcomes as well as lotteries in terms of ideal or permissible preferences.} 
} 

\textcolor{black}{The claim that the value of lotteries is grounded in the value of outcomes, then, means that \emph{reasons to prefer} one lottery over another must be grounded in the values (and probabilities) of their potential outcomes. 
In forming her preferences between lotteries, an ideally rational agent will be fundamentally concerned with \emph{how well things actually turn out}.} 
Of course, under conditions of uncertainty, the agent doesn't \emph{know} how things will actually turn out. What she does know is how things \emph{might} turn out if a given lottery is chosen, and with what probabilities. The fact that a lottery might turn out a certain way carries normative weight, from her perspective, because it corresponds to the \emph{possibility} of final value.

\textcolor{black}{More specifically, the fact that lottery $\lotOneQuant$ might yield outcome $o_1$, and lottery $L_2$ might yield outcome $o_2$, where $o_1$ is better than $o_2$, constitutes a \emph{pro tanto} requiring reason to prefer $\lotOneQuant$ over $L_2$. (The strength of this reason increases---not necessarily linearly---with the probability that $\lotOneQuant$ would yield $o_1$, the probability that $L_2$ would yield $o_2$, and the difference in value between $o_1$ and $o_2$. If the ranking of outcomes is incomplete, then other evaluative relations between $o_1$ and $o_2$, like parity, might provide \emph{justifying} reasons in favor of preferring either lottery, but not requiring reasons.}\footnote{\textcolor{black}{For exposition of the distinction between the justifying and requiring strength of a reason, see \cite{gert2003requiring}. By ``requiring reason'', in the main text, we mean any reason with at least some requiring strength; by ``justifying reason'', we mean a reason with justifying strength but no requiring strength.}}) \textcolor{black}{If no possible outcome of $L_1$ is better than any possible outcome of $L_2$, then there are no \textit{pro tanto} requiring reasons to prefer $L_1$ over $L_2$, so there cannot be all-things-considered requiring reason to prefer it. So, by the ideal preference analysis of the value of lotteries, $L_1$ cannot be better than $L_2$.}

\textcolor{black}{What the above line of reasoning rules out is that probabilistic combinations or averages of multiple outcomes could constitute \emph{irreducible} reasons to prefer one lottery over another---that is, none of the potential outcomes of $\lotOneEx$ on its own provides any requiring reason to prefer $\lotOneEx$ over $L_2$, but when we combine or average those outcomes together, a requiring reason emerges. That can't happen, we claim, because what ultimately matters is how things actually turn out, and things will actually turn out as one particular outcome, not a probabilistic combination or average of outcomes.} 

\textcolor{black}{We are attracted to the preference-based analysis of the value of lotteries, and think it provides especially strong support for Negative Dominance. But we emphasize that this particular picture is not essential to the more general line of reasoning that favors Negative Dominance. In short: If what ultimately matters in comparing 
lotteries is how well things actually turn out, then the comparative 
value of any two lotteries must be grounded in and explicable in terms of the comparative 
value of the ways each lottery \emph{might} turn out---that is, in one-to-one evaluative comparisons between their possible outcomes.}

It is instructive to contrast Negative Dominance with a stronger principle that has received substantial attention and support in the recent literature. This principle was originally proposed (though not endorsed) by Caspar Hare (\citeyear{hare2010take}). 
Hare works in a slightly different setting than ours, with two different primitive objects: states and prizes. In this setting, probabilities are defined on states, while betterness is defined in the first instance on prizes. \emph{Gambles} are functions from states to prizes. If we think of outcomes as state-prize pairs, then gambles \emph{induce} lotteries, since a probability distribution on states yields, via the gamble, a probability distribution on state-prize pairs. But gambles have more structure than lotteries, allowing us to contrast a possible outcome of a given gamble with what \textit{would have} happened (in the same state) if a different gamble had been chosen. 

The key idea in Hare's principle can be stated as:
\begin{description}
\item[Statewise Negative Dominance] 
\textcolor{black}{If there is no state $s$ in which gamble $G_1$ yields a better prize than gamble $G_2$, then $G_1$ is not better than $G_2$.}\footnote{Hare calls his version of this principle ``Recognition'' (241).}
\end{description}

\textcolor{black}{Though Hare himself ultimately does not endorse this principle, many others in the subsequent literature have endorsed it, including \cite{schoenfield2014decision}, \cite{bales2014decision}, and \cite{doodyms}.} 
If we understand outcomes as state-prize pairs, as above, then Statewise Negative Dominance is strictly stronger than Negative Dominance. \textcolor{black}{So anyone who accepts Statewise Negative Dominance must accept our weaker principle as well.}

\textcolor{black}{That said, there are reasons to doubt Statewise Negative Dominance.} In particular, as \citet{bader2018stochastic} has emphasized, when the betterness relation on outcomes is incomplete, Statewise Negative Dominance conflicts with another, more widely accepted dominance principle:\footnote{ \citet[310]{manzini2008representation} earlier independently presented a similar argument for an incompatibility between Stochastic Dominance (their ``Preference Sure Thing'') and their ``Vague Sure Thing Principle'' (VST). Their VST is even stronger than Statewise Negative Dominance, but the argument for incompatibility is the same.}
\begin{description}
\item[Stochastic Dominance] (i) If for every outcome $o$, $\lotOneQuant$ has at least as great a probability as $L_2$ of yielding an outcome at least as good as $o$, then $\lotOneQuant$ is at least as good as $L_2$. (ii) If in addition, for some $o$, $\lotOneQuant$ has a strictly greater probability of yielding an outcome at least as good as $o$, then $\lotOneQuant$ is better than $L_2$.\footnote{\textcolor{black}{\citet{russellfixing} shows that if the betterness ordering can be incomplete, or can have certain non-standard infinitary structures, then this standard formulation of Stochastic Dominance gives the wrong results and must be revised.} But the difference between the standard formulation and Russell's proposed revision won't matter here.}
\end{description}
When condition (i) is satisfied, we say that $\lotOneQuant$ \emph{weakly stochastically dominates} $L_2$. When condition (ii) is also satisfied, we say that $\lotOneQuant$ \emph{strictly stochastically dominates} (or simply \emph{stochastically dominates}) $L_2$.

Stochastic Dominance is extremely plausible, and among the most widely accepted principles in normative decision theory. (For arguments in its favor, \textcolor{black}{see \citet[\S 3.4.4]{easwaran2014decision}, 
\cite[\S 4]{russellfixing}, and \citet[\S4]{tarsneyFCexpected}}.) It is compatible with a very wide range of attitudes toward risk, being satisfied not just by standard expected utility theory but also by more permissive alternatives like rank-dependent/risk-weighted expected utility \citep{quiggin1982theory,buchak2013risk} and the version of weighted-linear utility theory recently defended by \citet{bottomley2024rational}.\footnote{To put the point more exactly, a classic result from \citet{hadar1969rules} and \citet{hanoch1969efficiency} implies that an expected utility maximizer will satisfy Stochastic Dominance with respect to a given ranking of outcomes if her utility function represents that ranking, i.e., assigns greater utility to one outcome than another iff the first outcome is higher in the ranking.
The same thing turns out to be true of the alternatives to expected utility theory mentioned in the main text.} 
\textcolor{black}{The fact that Statewise Negative Dominance is incompatible with Stochastic Dominance in the presence of incomparability therefore substantially undercuts the \emph{prima facie} plausibility of Statewise Negative Dominance for anyone not already committed to completeness.}

\textcolor{black}{Some will nevertheless respond to this argument by giving up Stochastic Dominance in favor of Statewise Negative Dominance. Others, like us, will take it to show that \textcolor{black}{anyone who believes in evaluative incomparability has good reason to reject Statewise Negative Dominance}. 
But in either case, crucially, the preceding argument does not give any reason to reject Negative Dominance. First, Negative Dominance is consistent with Stochastic Dominance even in the presence of incompleteness (\citeauthor{ledermanmultidimensional} \citeyear{ledermanmultidimensional}: Proposition 3.1; \citeauthor{ledermanmarbles} \citeyear{ledermanmarbles}: \S 7). 
\textcolor{black}{Second, in our view, Negative Dominance is more compellingly motivated than Statewise Negative Dominance. In particular,} our earlier positive argument for Negative Dominance, which turned on grounding comparisons between lotteries in the final value of their potential outcomes, does not support Statewise Negative Dominance and is therefore not undermined by its rejection.} \textcolor{black}{Statewise Negative Dominance insists that comparisons between lotteries must be grounded in comparisons between their outcomes \emph{in the same state of nature}. But nothing in our argument for Negative Dominance motivates this restriction.}\footnote{This claim may seem surprising given that \cite{schoenfield2014decision} and \citet{doodyms} have suggested that Statewise Negative Dominance can be motivated by a primary concern for 
\textcolor{black}{the value of outcomes}. 
But this concern does not justify an insistence on statewise comparisons.
Instead, this insistence represents (in our view) a mistaken concern with \emph{counterfactuals}---with comparing the actual outcome of one's choice to \emph{what would have happened, if one had made a different choice in the same state of nature}. There is no particular reason, insofar as one is \textcolor{black}{ultimately} concerned with \textcolor{black}{\textit{how things actually turn out}}, 
to privilege this counterfactual outcome over all the other possible unrealized outcomes of the gambles under consideration (which are all equally non-actual).}

\section{Personal Good}
\label{s:personalGood}

How should the relation of personal betterness on individual lives connect to the relation of impersonal betterness on outcomes and lotteries? The following principle articulates a weak constraint:
\begin{description}
\item[Pareto] If the same individuals exist in outcomes $o_1$ and $o_2$, and $o_1$ is personally at least as good as $o_2$ for every individual and personally better for at least one individual, then $o_1$ is impersonally better than $o_2$.\footnote{\textcolor{black}{\textcolor{black}{This principle is most plausible} on the assumption that there are no sources of impersonal value in the world besides the personal value of individual lives. Pareto could reasonably be rejected, for instance, by someone who believes that natural beauty is a source of impersonal value that can outweigh differences in welfare or other kinds of personal value. Throughout this paper we make the broadly welfarist assumption that the personal value of individual lives is indeed the sole source of impersonal value. Those who reject that assumption may wish to understand our arguments as concerning not impersonal value but some narrower concept like ``welfarist value''. 
This understanding of our arguments could also sidestep egalitarian and other ``holistic'' worries about Personal Good, a topic to which we'll come shortly. \textcolor{black}{(\citet[\S 3.1]{thomas2023asymmetry} makes a similar point in defending his principle of ``Person-Based Choice'' against egalitarian objections.)}}\label{welfarism}}
\end{description}

This principle only concerns \emph{outcomes}. How might it be extended to lotteries? We will work with the following weak extension:
\begin{description}
\item[Personal Good] Suppose 
every possible individual $i$ has the same probability of existence in lottery $\lotOneQuant$ as in lottery $L_2$. Then, (i) if for every individual $i$ and life $l$, $\lotOneQuant$ gives $i$ at least as great a probability of a life at least as personally good as $l$, compared to $L_2$, then $\lotOneQuant$ is at least as impersonally good as $L_2$; and (ii) if in addition, for some $i$ and $l$, $\lotOneQuant$ gives $i$ a strictly greater probability of a life at least as personally good as $l$, then $\lotOneQuant$ is impersonally better than $L_2$.\footnote{Our arguments below will rely mainly on part (ii) of Personal Good, appealing to part (i) only in \S \ref{s:lotteries}. Personal Good (especially part (i)) is closely related to the principle that \citet{mccarthy2020utilitarianism} call ``Anteriority''.}
\end{description}

This principle extends Pareto to situations of risk, and to situations where different individuals may exist in different outcomes. But it is very modest in both respects. First, it only covers situations where every individual has a fixed probability of existence, saying nothing about choices that affect that probability for even a single individual. 
This means that it is compatible with a wide range of views about how to value nonexistence, including the view that existence and nonexistence are always incomparable. Second, within that restricted setting, Personal Good reflects 
only a minimal claim about the conditions under which one lottery is better than another from a personal perspective: \textcolor{black}{$L_1$ is at least as good as $L_2$ from an individual's perspective if the \emph{personal lottery} (probability distribution over lives) that she receives from $L_1$ weakly stochastically dominates the personal lottery she receives from $L_2$; and it is strictly better if the stochastic dominance relation is strict}. 
As described above, Stochastic Dominance is a minimal and compelling principle for ranking risky alternatives, which is compatible with a very wide range of attitudes toward risk.

\textcolor{black}{That said, there is precedent for rejecting Personal Good. In particular, some egalitarians will reject Personal Good because, while it allows concern for \textit{ex ante} equality (equality between personal lotteries), it is incompatible with natural ways of spelling out an independent concern for \textit{ex post} equality (equality between individual lives in a given outcome) \citep{myerson1981utilitarianism}. For instance, consider the following pair of lotteries:}
\begin{table}[!htb]

\centering

\begin{minipage}{.3\linewidth}
\begin{subtable}{\linewidth}
\centering
\caption{\textsc{Lottery 1}}
\begin{tabular}{lcc}
	\toprule
	& Heads & Tails \\
	\midrule
	$P_1$	& $1$ & $0$ \\
	$P_2$	& $1$   & $0$ \\	
	\bottomrule
\end{tabular}
\end{subtable}

\end{minipage}
\begin{minipage}{.3\linewidth}
\begin{subtable}{\linewidth}
\centering
\caption{\textsc{Lottery 2}}
\begin{tabular}{lcc}
	\toprule
	& Heads & Tails \\
	\midrule
	$P_1$	& $1$ & $0$ \\
	$P_2$	& $0$	& $1$ \\	
	\bottomrule
\end{tabular}
\end{subtable}		
\end{minipage}

\end{table}

\noindent \textcolor{black}{Personal Good implies that these two lotteries are equally good. But since Lottery 1 guarantees perfect \emph{ex post} equality, while Lottery 2 guarantees some \emph{ex post} inequality, one might judge that Lottery 1 is better. Relatedly, someone who places value on \textit{ex post} \emph{diversity} 
might question Personal Good: If we replace ``1'' and ``0'' in the table above with two lives that are different but equally good, one might judge that Lottery 2 is better since it guarantees diversity while Lottery 1 guarantees bland uniformity.}

\textcolor{black}{We will set these objections aside, however, because (as we'll explain in the next section) our impossibility results can be established using examples where the lotteries under comparison are intuitively equivalent in terms of both \textit{ex post} equality and \textit{ex post} diversity. Thus the central argument of this paper, which is built on those results, could be run using a restricted version of Personal Good that only applies when neither equality nor diversity is at stake.}

\section{Incomparable lives} 
\label{s:resultsForPersonalIncompleteness}

We now turn to the central arguments of the paper, showing that the combination of Negative Dominance and Personal Good rules out many forms of value incompleteness. We begin in this section by considering incompleteness in the personal betterness ranking of lives, along with those forms of incompleteness in the impersonal betterness ranking of outcomes that arise from it.

\textcolor{black}{Our arguments will require some auxiliary premises, which we now introduce.} We first make an assumption that we think of as basically a convenience, without significant conceptual import: 

\begin{description}
\item[Richness] For any finite set of possible individuals and any function from those individuals to lives other than non-existence, there is an outcome where exactly those individuals exist, with exactly those lives. Moreover, for any finite set of outcomes and any probability distribution defined on that set, there is a lottery which assigns those probabilities to those outcomes.
\end{description}

We are assuming (as we said earlier) that the set of possible individuals is countably infinite, and that the impersonal betterness relation is defined over the whole set of lotteries. 
We won't need anything like the full strength of Richness; all we  need is that the domain of the betterness relation includes outcomes with the structure of our examples below.

Next, we make an assumption that connects incomparability in personal betterness to incomparability in impersonal betterness:
\begin{description}
\item[Incomparability Transmission] Suppose that outcomes $o_1$ and $o_2$ are identical except that one possible individual receives life $a$ in $o_1$ and life $b$ in $o_2$, where $a$ and $b$ are incomparable in terms of personal value. Then $o_1$ and $o_2$ are incomparable in terms of impersonal value.
\end{description}

This principle is a natural counterpart of Pareto in the context of personal incomparability: 
We compare outcomes by comparing them from the vantage point of each individual, and when only one individual's fate differs between two outcomes, the impersonal betterness ordering should reflect that individual's personal betterness ordering. When the two lives that individual might receive are personally incomparable, their incomparability is ``transmitted'' to the outcomes in which they figure.\footnote{\textcolor{black}{Like 
Personal Good, Incomparability Transmission is open to objections related to the impersonal value of equality and diversity. 
We consider those objections later in this section.}}

Finally, we will assume the following principle about the structure of incomparability on outcomes: 

\begin{description}
\item[Small Improvements] Suppose that outcome $o_1$ is impersonally incomparable with outcomes $o_2$ and $o_3$. Then there is either \begin{enumerate}
\item[(i)] an outcome $o_1^+$ which is better than $o_1$ for some individual, and equally good for everyone else, while still being impersonally incomparable with both $o_2$ and $o_3$, or
\item[(ii)] an outcome $o_1^-$ which is worse than $o_1$ for some individual, and equally good for everyone else, while still being impersonally incomparable with both $o_2$ and $o_3$.
\end{enumerate}
\end{description}

Insensitivity either to small improvements or to small worsenings is a paradigmatic feature of incomparability, by contrast to equal goodness. Thus it should not be controversial that, given incomparable $o_1$ and $o_2$, there is either an $o_1^+$ that improves $o_1$ and is still incomparable with $o_2$, or an $o_1^-$ that worsens $o_1$ and is still incomparable with $o_2$. The above principle adds to this uncontroversial claim two further ideas. First, it adds that we can achieve this improvement (worsening) by improving (worsening) the life of some individual in $o_1$. This is slightly stronger, but we do not expect it to be particularly controversial 
\textcolor{black}{(especially given our broadly welfarist assumption that the only evaluatively significant features of an outcome are which individuals exist and what lives they lead, see n. \ref{welfarism})}. Second, it adds that when $o_1$ is incomparable to \emph{two} further outcomes $o_2$ and $o_3$, there is a single improvement (or worsening) which preserves incomparability with \emph{both} $o_2$ and $o_3$. For instance, if an outcome where Patricia is an artist is incomparable with an outcome where she is a banker and an outcome where she is a CIA operative, then we can either slightly improve or slightly worsen the first outcome (e.g., by slightly adjusting Patricia's income) while preserving both incomparabilities. Like the first, this second claim seems plausible to us, and holds in simple, natural models of value incomparability. In any case, as described below, our argument makes only limited use of Small Improvements, and does not require it to hold in full generality.

Now we can state our first result:

\begin{proposition} 
Given Richness, the principles of Negative Dominance, Personal Good, Incomparability Transmission, and Small Improvements  rule out \textcolor{black}{personal incomparability between lives}.\label{thm:personalIncomp}
\end{proposition}

\textcolor{black}{This can be proved using the example from the introduction, restated for ease of reference} in Table \ref{tab:personalIncomp}.
\begin{table}[!htb]

\centering

\begin{minipage}{.3\linewidth}
\begin{subtable}{\linewidth}
	\centering
	\caption{\textsc{Lottery 1}}
	\begin{tabular}{lcc}
		\toprule
		& Heads & Tails \\
		\midrule
		$P_1$	& $a$ & $b$ \\
		$P_2$	& \textcolor{black}{$a$}   & \textcolor{black}{$b$} \\	
		\bottomrule
	\end{tabular}
\end{subtable}

\end{minipage}
\begin{minipage}{.3\linewidth}
\begin{subtable}{\linewidth}
	\centering
	\caption{\textsc{Lottery 2}}
	\begin{tabular}{lcc}
		\toprule
		& Heads & Tails \\
		\midrule
		$P_1$	& $a^+$ & $b$ \\
		$P_2$	& \textcolor{black}{$b$}	& \textcolor{black}{$a$} \\	
		\bottomrule
	\end{tabular}
\end{subtable}		
\end{minipage}

\caption{}
\label{tab:personalIncomp}  
\end{table}

\noindent Again, we have two lotteries affecting two possible individuals, $P_1$ and $P_2$. (We can assume that everyone else is assured of non-existence.) Given two 
incomparable lives $a$ and $b$, Incomparability Transmission implies that the outcomes $(a,a)$ and $(b,b)$ are both incomparable with the outcomes $(a,b)$ and $(b,a)$. 
Small Improvements implies that we can find an outcome that is equally good for one individual and better for the other than $(a,b)$ while still being incomparable with both $(a,a)$ and $(b,b)$, or else an outcome that is equally good for one individual and worse for the other, while still being incomparable with both $(a,a)$ and $(b,b)$. Let's concentrate on the first case, and assume the improvement is to the first individual, giving this person the life $a^+$. (The other cases are very similar.) The outcome $(a^+,b)$ is by assumption still incomparable with $(a,a)$ and $(b,b)$. Richness then 
guarantees that the domain of the impersonal betterness relation includes the two lotteries described in Table \ref{tab:personalIncomp}. Both possible outcomes of Lottery 2 are incomparable with both possible outcomes of Lottery 1, so Negative Dominance implies that Lottery 2 is not better than Lottery 1. But Lottery 2 gives $P_1$ a stochastically dominant personal lottery, and $P_2$ a stochastically equivalent lottery, 
so Personal Good implies that Lottery 2 \emph{is} better than Lottery 1.\footnote{\citet{nebel2020fixed} also presents an opaque sweetening-inspired puzzle in a population-ethical context. His puzzle compares the outcomes $(a,b)$ and $(b^+,a^+)$, where $a$ and $a^+$ are both incomparable with $b$ and $b^+$. He points out that the latter outcome is intuitively better than the former, and that this verdict can be derived from Pareto, anonymity, and transitivity, but that this conflicts with the ``Person-Affecting Restriction’’ (PAR), which claims that one outcome is better than another only if there is someone for whom it's better. Nebel sees this as a reason to give up PAR, and we agree. But giving up PAR does not resolve our puzzle, which uses only the much weaker and more plausible principle of Incomparability Transmission. (Unlike PAR, Incomparability Transmission is consistent with the other principles Nebel  invokes even in the presence of personal incomparability.)

As Nebel does to Hare's argument, it's also possible to generate a structural analogue of our argument that inverts the roles of states and individuals. In this inverted argument, the analogue of Negative Dominance asserts that \emph{if no individual's personal lottery under $L_1$ is better than any individual's personal lottery under $L_2$, then $L_1$ is not better than $L_2$}, while the analogue of Personal Good asserts that \emph{if in at least one state, the outcome of $L_1$ is Suppes-Sen superior to the outcome of $L_2$ (i.e., Pareto superior to a permutation of that outcome), and in all other states, the outcomes of $L_1$ and $L_2$ are identical, then $L_1$ is better than $L_2$}. (To put these principles into conflict, just switch the rows and columns of Lottery 1 in Table \ref{tab:personalIncomp}, with $b$'s in the top row and $a$'s in the bottom. The argument will of course require auxiliary assumptions analogous to those used in Proposition \ref{thm:personalIncomp}. \textcolor{black}{Note also that the second principle requires the introduction of states of nature, which the principles in the main text avoid. The identification of states of nature across different lotteries is analogous to the identification of individuals across different lotteries required by our principle of Personal Good.})
}

This argument establishes, we think, that Negative Dominance and Personal Good are jointly incompatible with the existence of incomparable lives. To briefly consider the other ways out: Richness is only used to secure the existence of two lotteries with the structure of Lottery 1 and Lottery 2 in Table \ref{tab:personalIncomp}, which are hardly exotic. Our use of Small Improvements should also be uncontroversial: given that $(a,b)$ is incomparable with $(a,a)$ and $(b,b)$, then it seems (in the spirit of insensitivity to small improvements) that it should be possible to slightly improve or worsen someone's life in $(a,b)$ while maintaining that incomparability.\footnote{\textcolor{black}{Recall from n.\ \ref{fn:superStrongPareto} that we can't judge $(a^+,b)$ to be better than $(a,a)$ on the grounds that it is better for someone and worse for no one, on pain of generating cycles of betterness.}} 
We think this principle is plausible in full generality, but, as with Richness, rejecting it in particular cases would not be enough \textcolor{black}{to block our argument, which only requires there to be} 
\emph{some} case with the structure of the one above.

\textcolor{black}{The most debatable of our auxiliary assumptions is Incomparability Transmission. This was used to secure the conclusion that, if $a$ and $b$ are personally incomparable, then $(a,a)$ and $(b,b)$ are impersonally incomparable with $(a,b)$ and $(b,a)$. This claim seems quite plausible on its face.   
But as with Personal Good, it might be questioned by someone who places intrinsic value on \textit{ex post} equality and/or diversity.\footnote{\textcolor{black}{Thanks to two anonymous reviewers for this point.}} As we saw in \S \ref{s:personalGood}, someone who values equality or diversity might reject Personal Good because 
it takes no account of correlations between individual prospects and hence is insensitive to holistic features of outcomes. Likewise, they might reject Incomparability Transmission on the grounds that, even if lives $a$ and $b$ are incomparable from an individual perspective, one of them might yield an overall outcome that is more equal or more diverse, and therefore impersonally better. For instance, someone who values \emph{ex post} diversity might claim that $(a,b)$ and $(b,a)$ are both impersonally better than $(a,a)$ and $(b,b)$, contra Incomparability Transmission.}

\textcolor{black}{But such equality- and diversity-based objections to our argument (whether directed at Personal Good, Incomparability Transmission, or both) can be defused by slightly modifying our central example, as illustrated in Table \ref{tab:equalityDiversityCase}:}

\begin{table}[!htb]

\centering

\begin{minipage}{.3\linewidth}
\begin{subtable}{\linewidth}
\centering
\caption{Lottery 1}
\begin{tabular}{lcc}
\toprule
& Heads & Tails \\
\midrule
$P_1$	& $a$ 	& $c$ \\
$P_2$	& $b$   & $d$ \\	

\bottomrule
\end{tabular}
\end{subtable}

\end{minipage}
\begin{minipage}{.3\linewidth}
\begin{subtable}{\linewidth}
\centering
\caption{Lottery 2}
\begin{tabular}{lcc}
\toprule
& Heads & Tails \\
\midrule
$P_1$	& $a^+$ & $c$ \\
$P_2$	& $d$	& $b$ \\	
\bottomrule
\end{tabular}
\end{subtable}		
\end{minipage}

\caption{}
\label{tab:equalityDiversityCase}  
\end{table}

\noindent \textcolor{black}{In this new example, there are four pairwise incomparable lives, with $P_1$ facing a lottery between $a$ and $c$, and $P_2$ facing a lottery between $b$ and $d$. As before, Lottery 2 reverses the correlations between the two personal lotteries while sweetening life $a$ for $P_1$. But now there is no longer any apparent difference between the possible outcomes of the two lotteries in terms of either equality or diversity.}\footnote{\textcolor{black}{One might question whether \emph{any} two outcomes consisting of two incomparable lives must be equivalent in terms of equality and diversity. But again, we don't need this to be true in full generality. And it seems very plausible that, assuming  personal incompleteness and given systematic theories of equality and diversity under incompleteness, we can find \emph{some} quadruple of incomparable lives that would make the outcomes in Table \ref{tab:equalityDiversityCase} equivalently equal and diverse.}} 
\textcolor{black}{Using this example, then, our argument would go through even if we restricted both Personal Good and Incomparability Transmission to contexts where neither equality nor diversity is at stake. 
We haven't made this example the centerpiece of our discussion because it requires assuming the existence of four pairwise incomparable lives, which is not 
guaranteed by personal incompleteness \emph{per se}. But most natural incomplete theories of personal betterness will satisfy this assumption. 
For instance, it would be strange to think that \emph{artist} and \emph{banker} are incomparable but to deny that we can find two more lives (say, \emph{CIA operative} and \emph{dairy farmer}) to complete the required quartet.}

A notable special case of the preceding result is where one of the incomparable ``lives'' is nonexistence.\footnote{\textcolor{black}{Nonexistence incomparability is one form of incomplete personal betterness that might not allow us to find four pairwise incomparable lives: one might think that existence is incomparable with nonexistence, but that any two lives other than nonexistence can be compared. But in this} \textcolor{black}{case, worries about equality and diversity don't arise for our initial argument: If one of $a$ or $b$ is nonexistence, then none of the four outcomes in the original example involve any \emph{ex post} inequality or diversity.}} 
Many philosophers accept the following principle:

\begin{description}
\item[Variable-Population Incompleteness] There is some life $a$ such that, whenever two outcomes $o_1$ and $o_2$ are identical except that one possible individual does not exist in $o_1$ and exists with life $a$ in $o_2$, $o_1$ and $o_2$ are incomparable.
\end{description}

One motivation for this principle is the belief that existence and non-existence are incomparable in terms of \emph{personal} value, coupled with a principle like Incomparability Transmission. This might motivate the view that populations of different size are \emph{always} incomparable (\citeauthor{bader2022person} \citeyear{bader2022person,bader2022asymmetry}). Alternatively, whatever one thinks about the personal value of nonexistence, one might accept Variable-Population Incompleteness to avoid undesirable conclusions about impersonal value. For instance, critical-range utilitarianism, which holds that adding an individual with slightly positive welfare to a population results in incomparability, can avoid both the Repugnant Conclusion \citep[\S 131, 387--90]{parfit1984rap} and the Sadistic Conclusion \citep{arrhenius2000impossibility}.\footnote{\label{gustafssonnote}\cite{gustafsson2020population} proposes a version of critical-range utilitarianism that is motivated by a fourth category of value called ``undistinguishedness'' and shows that it escapes versions of the Repugnant and Sadistic Conclusions.} But the preceding argument shows us that both these views are in tension with Negative Dominance and Personal Good, taken together:

\begin{proposition} 
Given Richness, the principles of Negative Dominance, Personal Good, and Small Improvements rule out Variable-Population Incompleteness.
\end{proposition}

To see this, let $b$ in Table \ref{tab:personalIncomp} represent nonexistence and $a$ be a life of the sort described by Variable-Population Incompleteness, such that adding someone to the population with life $a$ yields incomparability. Then Variable-Population Incompleteness implies that $(a,a)$ and $(b,b)$ are incomparable with $(a,b)$ and $(b,a)$.\footnote{While the argument as stated appeals to the ``empty population'' $(b,b)$, we can avoid this by  simply adding another person living life $c$ to all outcomes under consideration.} 
The rest of the proof then follows the proof of Proposition \ref{thm:personalIncomp}. (Note that this result does not use Incomparability Transmission.)
The upshot is that Negative Dominance and Personal Good together tell strongly against views that see changes in population size as a source of incomparability.

\section{Incomparable outcomes} 
\label{s:resultsForPersonalCompleteness}

In the last section, we focused on incompleteness in the personal ranking of lives, and those forms of incompleteness in the impersonal ranking of outcomes that arise from it. But incomparability between outcomes can have other sources as well. For example, those who resist interpersonal tradeoffs or aggregation might hold that, in some or all cases where $o_1$ is better than $o_2$ for some individuals and worse for others, $o_1$ and $o_2$ are incomparable. Or one might hold that there are imprecise ``exchange rates'' between different desirable features of an outcome, like total vs.\ average welfare or average welfare vs.\ equality, so that pairs of outcomes that trade off these dimensions are sometimes incomparable. (Critical-range utilitarianism, discussed above, can be seen as a view of this kind, in which ``small population size'' is a desideratum that trades off imprecisely with total welfare.)

This section will show that, as long as population-level incomparability has a plausible ``separability'' property analogous to Incomparability Transmission, 
these views are vulnerable to essentially the same difficulty that we explored in the last section. In other words, Negative Dominance and Personal Good support a general argument against incompleteness in the impersonal value of outcomes, whether or not it arises from incompleteness in personal value.

Our argument will roughly parallel the argument in the last section, with \emph{subpopulations} taking the place of single individuals. A subpopulation is a set of possible individuals that have been arbitrarily indexed (i.e., assigned positive integers from 1 up to the number of individuals in the subpopulation). We also introduce the concept of a \emph{distribution}, which stands in approximately the relation to subpopulations that lives stand in to individuals. A distribution is a function from the first $n$ positive integers to the set of lives (including nonexistence). The \emph{size} of a subpopulation is the number of possible individuals it contains, and the size of a distribution is the number of lives it assigns. For any subpopulation $S$ consisting of possible individuals \textcolor{black}{$P_{S1}, P_{S2}, \ldots, P_{Sn}$}, and any distribution $D$ of size $n$ or less, we say that $S$ \emph{realizes} $D$ when $P_{S1}$ experiences the life $D(1)$, $P_{S2}$ experiences the life $D(2)$, and so on, with any leftover individuals in $S$ being assigned nonexistence.

We then give the following definition:

\begin{description}
\item[Strong Incomparability] Distributions $A$ and $B$ are \emph{strongly incomparable} if \textcolor{black}{and only if}, whenever outcomes $o_1$ and $o_2$ are identical except that some subpopulation realizes distribution $A$ in outcome $o_1$ and distribution $B$ in outcome $o_2$, $o_1$ is incomparable with $o_2$.
\end{description}

To say that $A$ and $B$ are strongly incomparable, in other words, is analogous to saying that lives $a$ and $b$ are personally incomparable in a way that satisfies Incomparability Transmission: exchanging one for the other, while holding everything else fixed, always generates incomparability.\footnote{Note, however, that we have not formally introduced a betterness ordering on distributions. To say that $A$ and $B$ are strongly incomparable does not mean that they are impersonally incomparable in our formal sense, because the latter relation is only defined on outcomes.}

It seems plausible that, whenever two outcomes are incomparable, the two distributions they realize (relative to a given indexing of the set of individuals who exist in either outcome) are strongly incomparable. 
But this is a substantial assumption, amounting to a form of separability: it allows us, in comparing two outcomes, to focus only on those individuals whose lives differ between the outcomes and ignore everyone else. Since this assumption may be controversial, we will focus on the much weaker claim that, if impersonal betterness is incomplete, then at least \emph{some} distributions are strongly incomparable. For example, perhaps a sufficiently sharp tradeoff between total welfare and inequality \textcolor{black}{in a subpopulation of a given size} always generates incomparability, 
even if for less sharp tradeoffs, the impersonal betterness ranking can depend on features of the unaffected population. As we will see, our argument will not require the full force of even this modest existential claim.

We now come to our next result:

\begin{proposition}
Given Richness, the principles of Negative Dominance, Personal Good, and Small Improvements rule out strong incomparability.
\end{proposition}

\begin{proof}
Suppose that $A$ and $B$ are strongly incomparable. 
Letting $n$ be the size of the larger of $A$ and $B$, choose two disjoint subpopulations $S = \{P_{S1}, P_{S2}, \ldots, P_{Sn}\}$ and $T = \{P_{T1}, P_{T2}, \ldots, P_{Tn}\}$. Now, consider the two lotteries in Table \ref{tab:strongIncomp}.

\begin{table}[!htb]

\centering

\begin{minipage}{.3\linewidth}
\begin{subtable}{\linewidth}
\centering
\caption{\textsc{Lottery 1}}
\begin{tabular}{lcc}
	\toprule
	& Heads & Tails \\
	\midrule
	$S$	& $A$ 	& $B$ \\
	$T$	& $\textcolor{black}{A}$   & $\textcolor{black}{B}$ \\	
	\bottomrule
\end{tabular}
\end{subtable}

\end{minipage}
\begin{minipage}{.3\linewidth}
\begin{subtable}{\linewidth}
\centering
\caption{\textsc{Lottery 2}}
\begin{tabular}{lcc}
	\toprule
	& Heads & Tails \\
	\midrule
	$S$	& $A^+$ & $B$ \\
	$T$	& $\textcolor{black}{B}$	& $\textcolor{black}{A}$ \\	
	\bottomrule
\end{tabular}
\end{subtable}		
\end{minipage}

\caption{}
\label{tab:strongIncomp}  
\end{table}

These are exact analogues of the two lotteries used to prove Proposition \ref{thm:personalIncomp}, with individuals $P_1$ and $P_2$ replaced by subpopulations $S$ and $T$, and lives $a$ and $b$ replaced by distributions $A$ and $B$. Each outcome assigns a distribution to each subpopulation: in the Tails outcome of Lottery 2, for instance, $P_{S1}$ receives $B(1)$, $P_{S2}$ receives $B(2)$, etc, and likewise $P_{T1}$ receives $A(1)$, $P_{T2}$ receives $A(2)$, etc, with any leftover individuals being assigned nonexistence.

Now, it follows from the definition of strong incomparability that $(A,A)$ and $(B,B)$ are both incomparable with $(A,B)$ and $(B,A)$ (where the first and second terms of each ordered pair represent the distributions realized by $S$ and $T$ respectively). By Small Improvements, there exists either an outcome $(A^+,B)$ or $(A, B^+)$ that slightly improves one life in $(A,B)$, or else an $(A^-,B)$ or $(A,B^-)$ that slightly worsens one life in $(A,B)$, that is still incomparable with both $(A,A)$ and $(B,B)$. As before, we focus on the first case. Richness guarantees that the domain of the impersonal betterness relation includes the two lotteries in Table \ref{tab:strongIncomp}. And in comparing these two lotteries, Negative Dominance implies that Lottery 2 is not better than Lottery 1, while Personal Good implies that it is.

\end{proof}

As before, 
all our argument really requires is that we can find distributions $A$, $A^+$, and $B$ such that $(A,A)$ and $(B,B)$ are both incomparable with $(A^+,B)$ and $(B,A)$. We think it will be hard for proponents of incomparability between outcomes to deny this claim (we ourselves cannot think of any plausible incomplete outcome axiology that does so), even if they reject our Small Improvements assumption or deny 
that any distributions are strongly incomparable. 
It seems hard, therefore, to maintain that the impersonal betterness ranking of outcomes is incomplete without denying either Negative Dominance or Personal Good.

\section{From outcomes to lotteries} 
\label{s:lotteries}

In the last section, we presented an argument for completeness in the impersonal betterness ranking of outcomes. But the impersonal betterness relation, we have assumed, is defined not just on outcomes but on lotteries. And a complete ordering of outcomes need not entail a complete ordering of lotteries. For instance, even if outcomes are completely ordered by total welfare, the ordering of lotteries might be incomplete because it supervaluates over a range of reasonable risk attitudes. At an extreme, for instance, lotteries might be ordered only by Stochastic Dominance \citep{tarsney2020exceeding}.

In this section, we take our argument a step further by showing that Personal Good goes a long way toward bridging the gap between completeness for outcomes and completeness for lotteries.
Thus, if we accept completeness for outcomes on grounds that include Personal Good, there is at least significant reason to accept completeness for lotteries as well. This argument has somewhat limited scope (limited to what we will call ``finite rational lotteries''), and it depends more heavily than the arguments of the last two sections on assumptions that go beyond Negative Dominance and Personal Good. So we do not think that this pair of principles support completeness for lotteries as strongly as they support completeness for lives and outcomes. But as we will see, the auxiliary assumptions we require are still quite plausible.

Those assumptions consist of Stochastic Dominance and two new principles:

\begin{description}
\item[Substitution] For any 
outcome $o$ and any sufficiently large subpopulation $S$ (where what counts as ``sufficiently large'' can depend on $o$\textcolor{black}{, but is always finite})
there is an outcome that is equally as good as $o$ in which only  individuals in $S$ exist.

\item[Transitivity] If $L_1$ is at least as good as $L_2$, and $L_2$ is at least as good as $L_3$, then $L_1$ is at least as good as $L_3$.  
\end{description}

Transitivity is a familiar principle that is very widely (though not universally) accepted, and we won't have anything new to say for or against it. Substitution says, intuitively, that any degree of impersonal value can be realized by any sufficiently large group of possible individuals. This principle might be denied, for instance, by those who think that populations in which different individuals exist are always incomparable. But against the background assumption of completeness for outcomes, which we have just argued for, it's hard to see how one might deny it. Given completeness, in particular, Substitution follows from

\begin{description}
\item[Weak Anonymity] For any outcome $o$ and permutation $\pi$ of the set of possible individuals, $o$ is neither better nor worse than $\pi(o)$ (where $\pi(o)$ is the outcome in which $\pi(o)(x)=o(\pi(x))$.
\end{description}
Weak Anonymity seems like a minimal expression of impartiality. Given completeness for outcomes, Substitution can be understood as a still weaker sort of anonymity principle, which permits significant (though not unlimited) partiality between possible individuals.\footnote{\label{fn:infiniteSubstitution}\textcolor{black}{
The background restriction to finite-population outcomes is important here because, as with other anonymity principles, Substitution becomes much stronger and easier to deny in infinite-population contexts. 
For instance, consider an outcome $o$ in which a countably infinite set $T$ of possible individuals exist, all with lives that are maximally or near-maximally good (assuming that to be possible). And consider a subset $S$ of $T$ such that $S$ and its relative complement in $T$ are both infinite. Although $S$ is as large, in cardinal terms, as a subpopulation can get (given the assumption of only countably many possible individuals), it is at least doubtful whether any outcome in which only members of $S$ exist can be as good as $o$.}}

Our argument will be restricted to \emph{finite rational lotteries}: \textcolor{black}{lotteries over finite-population outcomes that are finitely supported (i.e., have only finitely many possible outcomes),}  
with each outcome having a rational probability. The finitary restrictions are substantive\textcolor{black}{, and we make them explicit here (although they have been implicitly in force throughout) for two reasons. First, this is the only one of our results that formally depends on finitary restrictions---in particular, the proof below requires the assumption of finite support. Second, while the result could formally be extended to finite-support lotteries over infinite populations, there is particular reason here to be skeptical of our axioms in that context---see n.\ \ref{fn:infiniteSubstitution}.} By contrast, we regard the ``rationality'' restriction as merely technical---it seems obvious that if lotteries with rational probabilities are completely ordered, then lotteries with real probabilities are as well.

Here is our next result:

\begin{proposition} 
If outcomes are completely ordered by impersonal betterness then, given Richness, 
the principles of Personal Good, Stochastic Dominance, Substitution, and Transitivity imply that the set of finite rational lotteries is also completely ordered by impersonal betterness.
\end{proposition}

\begin{proof}
Our strategy is to show that, given Personal Good, Stochastic Dominance, Substitution, and Transitivity, the ranking of finite rational lotteries ``supervenes on'' the ranking of outcomes: that is, for any finite rational lotteries $L_1$ and $L_2$, we can find a pair of outcomes $o_1$ and $o_2$ such that $L_1$ is at least as good as $L_2$ if and only if $o_1$ is at least as good as $o_2$.\footnote{We take inspiration here from \cite{thomas2023asymmetry}, who proves a similar supervenience result from a different but overlapping set of premises. (Roughly, Thomas does not assume Transitivity or Stochastic Dominance, but makes an anonymity assumption stronger than Substitution.)} This implies that, if the ranking of outcomes is complete, then the ranking of finite rational lotteries is as well.

Consider, then, an arbitrary pair of finite rational lotteries $L_1$ and $L_2$. To reduce the comparison between $L_1$ and $L_2$ to a comparison between outcomes, we will subject them to a series of transformations. First, we find the least common denominator of all the outcome probabilities in $L_1$ and $L_2$, which we will call $n$. We can then imagine $L_1$ and $L_2$ as gambles in a situation with $n$ equiprobable states. (As in our previous arguments, imagining lotteries as gambles is simply a useful presentational device---states don't play any  formal role in our argument.) Take the set of individuals who exist in any outcome of $L_1$ or $L_2$, index them arbitrarily, and call the resulting subpopulation $S_1$. $L_1$ and $L_2$ can then be represented as assigning distributions to this subpopulation in each state, as in Table \ref{tab:lotteriesStep1}.

\begin{table}[!htb]

\centering

\begin{minipage}{.49\linewidth}
\begin{subtable}{\linewidth}
\centering
\caption{\textsc{Lottery 1}}
\begin{tabular}{lcccc}
	\toprule
	& $s_1$ $\left( \frac{1}{n}\right)$ & $s_2$ $\left( \frac{1}{n}\right)$ & $\ldots$ & $s_n$ $\left( \frac{1}{n}\right)$  \\
	\midrule
	$S_1$	& $D_{1.1}$ 	& $D_{1.2}$ & $\ldots$ & $D_{1.n}$ \\
	\bottomrule
\end{tabular}
\end{subtable}

\end{minipage}
\begin{minipage}{.49\linewidth}
\begin{subtable}{\linewidth}
\centering
\caption{\textsc{Lottery 2}}
\begin{tabular}{lcccc}
	\toprule
	& $s_1$ $\left( \frac{1}{n}\right)$ & $s_2$ $\left( \frac{1}{n}\right)$ & $\ldots$ & $s_n$ $\left( \frac{1}{n}\right)$ \\
	\midrule
	$S_1$	& $D_{2.1}$ 	& $D_{2.2}$ & $\ldots$ & $D_{2.n}$ \\
	\bottomrule
\end{tabular}
\end{subtable}		
\end{minipage}

\caption{}
\label{tab:lotteriesStep1}  
\end{table}

The next step is to replace each outcome of each lottery with an equally good outcome, using ``fresh'' sets of individuals, so that no possible individual exists in more than one outcome of a given lottery. (It doesn't matter whether the set of individuals who might exist in one lottery overlaps with the set of individuals who might exist in the other.) Call the resulting lotteries $L_{1'}$ and $L_{2'}$. These transformed lotteries are illustrated in Table \ref{tab:lotteriesStep2}. Substitution (together with Richness) guarantees that we can find lotteries with the required properties, since for each outcome of each lottery, we can find an equally good outcome involving a fresh set of individuals. Stochastic Dominance guarantees that $L_{1'}$ and $L_1$ are equally good, and that $L_{2'}$ and $L_2$ are equally good.

\begin{table}[!htb]

\centering
\footnotesize

\begin{minipage}{.495\linewidth}
\begin{subtable}{\linewidth}
\centering
\caption{\textsc{Lottery $1'$}}
\begin{tabular}{lcccc}
	\toprule
	& $s_1$ $\left( \frac{1}{n}\right)$ & $s_2$ $\left( \frac{1}{n}\right)$ & $\ldots$ & $s_n$ $\left( \frac{1}{n}\right)$  \\
	\midrule
	$S_{1'.1}$	& $D_{1.1}$ 	& --- & $\ldots$ & --- \\
	$S_{1'.2}$	& --- 	& $D_{1'.2}$ & $\ldots$ & --- \\
	$\vdots$ & $\vdots$ & $\vdots$ & $\ddots$ & $\vdots$ \\
	$S_{1'.n}$	& --- 	& --- & $\ldots$ & $D_{1'.n}$ \\
	\bottomrule
\end{tabular}
\end{subtable}

\end{minipage}
\begin{minipage}{.495\linewidth}
\begin{subtable}{\linewidth}
\centering
\caption{\textsc{Lottery $2'$}}
\begin{tabular}{lcccc}
	\toprule
	& $s_1$ $\left( \frac{1}{n}\right)$ & $s_2$ $\left( \frac{1}{n}\right)$ & $\ldots$ & $s_n$ $\left( \frac{1}{n}\right)$  \\
	\midrule
	$S_{2'.1}$	& $D_{2.1}$ & --- & $\ldots$ & --- \\
	$S_{2'.2}$	& --- 	& $D_{2'.2}$ & $\ldots$ & --- \\
	$\vdots$ & $\vdots$ & $\vdots$ & $\ddots$ & $\vdots$ \\
	$S_{2'.n}$	& --- 	& --- & $\ldots$ & $D_{2'.n}$ \\
	\bottomrule
\end{tabular}
\end{subtable}		
\end{minipage}

\caption{}
\label{tab:lotteriesStep2}  
\end{table}

The final step is to ``stack'' all the individual outcomes other than nonexistence, in each lottery, into state $s_1$. Since all the possible outcomes of $L_{1'}$ involve disjoint populations, and likewise for $L_{2'}$, this stacking is possible. The result is a pair of lotteries $L_{1''}$ and $L_{2''}$ which guarantee an ``empty'' outcome in every state except $s_1$, while yielding a potentially massive population in that single state. These lotteries are illustrated in Table \ref{tab:lotteriesStep3}. Since all the states are equiprobable, this second transformation does not affect any individual's personal lottery, so Personal Good implies that $L_{1''}$ and $L_{1'}$ are equally good, and likewise that $L_{2''}$ and $L_{2'}$ are equally good.

\begin{table}[!htb]

\centering
\footnotesize

\begin{minipage}{.495\linewidth}
\begin{subtable}{\linewidth}
\centering
\caption{\textsc{Lottery $1''$}}
\begin{tabular}{lcccc}
	\toprule
	& $s_1$ $\left( \frac{1}{n}\right)$ & $s_2$ $\left( \frac{1}{n}\right)$ & $\ldots$ & $s_n$ $\left( \frac{1}{n}\right)$  \\
	\midrule
	$S_{1'.1}$	& $D_{1.1}$ 	& --- & $\ldots$ & --- \\
	$S_{1'.2}$	& $D_{1'.2}$ 	& --- & $\ldots$ & --- \\
	$\vdots$ & $\vdots$ & $\vdots$ & $\ddots$ & $\vdots$ \\
	$S_{1'.n}$	& $D_{1'.n}$ 	& --- & $\ldots$ & --- \\
	\bottomrule
\end{tabular}
\end{subtable}

\end{minipage}
\begin{minipage}{.495\linewidth}
\begin{subtable}{\linewidth}
\centering
\caption{\textsc{Lottery $2''$}}
\begin{tabular}{lcccc}
	\toprule
	& $s_1$ $\left( \frac{1}{n}\right)$ & $s_2$ $\left( \frac{1}{n}\right)$ & $\ldots$ & $s_n$ $\left( \frac{1}{n}\right)$  \\
	\midrule
	$S_{2'.1}$	& $D_{2.1}$ & --- & $\ldots$ & --- \\
	$S_{2'.2}$	& $D_{2'.2}$ 	& --- & $\ldots$ & --- \\
	$\vdots$ & $\vdots$ & $\vdots$ & $\ddots$ & $\vdots$ \\
	$S_{2'.n}$	& $D_{2'.n}$ 	& --- & $\ldots$ & --- \\
	\bottomrule
\end{tabular}
\end{subtable}		
\end{minipage}

\caption{}
\label{tab:lotteriesStep3}  
\end{table}

Let $o_1$ designate the outcome of $L_{1''}$ in state $s_1$, and $o_2$ designate the outcome of $L_{2''}$ in state $s_1$. Since $L_{1''}$ and $L_{2''}$ yield the same outcome in every state except $s_1$, Stochastic Dominance implies that, if $o_1$ is at least as good as $o_2$, then $L_{1''}$ is at least as good as $L_{2''}$; and likewise, if $o_2$ is at least as good as $o_1$, then $L_{2''}$ is at least as good as $L_{1''}$. Since we have assumed that outcomes are completely ordered, we know that at least one of these conditions obtains. Therefore, we know that either $L_{1''}$ is at least as good as $L_{2''}$ or vice versa (or both). Finally, Transitivity implies that $L_1$ and $L_{1''}$ are equally good, $L_{2}$ and $L_{2''}$ are equally good, and therefore that $L_{1}$ is at least as good as $L_{2}$ if and only if $L_{1''}$ is at least as good as $L_{2''}$. Thus either $L_1$ is at least as good as $L_2$, or $L_2$ is at least as good as $L_1$, or both. Since $L_1$ and $L_2$ were chosen arbitrarily, this shows that finite rational lotteries are completely ordered.

\end{proof}

Should a reader convinced by the arguments for outcome completeness in the last two sections be convinced by this argument as well? We think so. The additional premises are Stochastic Dominance, Substitution, and Transitivity. Stochastic Dominance is, as we have said, a very plausible and widely accepted principle. And, although it neither implies nor is implied by Negative Dominance, both principles reflect the basic idea that we should evaluate lotteries based on the values and probabilities of their possible outcomes, so the principles have some overlapping appeal. 
Transitivity is likewise widely accepted. And although Substitution could certainly be questioned, it is hard to reject if we have already accepted completeness for outcomes, since it is then entailed by Weak Anonymity, which is widely accepted for finite populations.

\section{Objections}
\label{s:rejectingNegativeDominance}

\textcolor{black}{The last three sections have shown that, against weak background assumptions, Negative Dominance and Personal Good are incompatible with 
incompleteness in the personal value of lives and the impersonal value of finite outcomes and lotteries. 
We see this as a compelling argument against those forms of incompleteness. But of course defenders of incompleteness might choose to reject Negative Dominance or Personal Good instead.}
\footnote{\textcolor{black}{Indeed, we should note that, while the authors of this paper all see Negative Dominance and Personal Good as  compelling principles, we ourselves are not uniformly certain of their correctness. In particular, one of us (CT) has some idiosyncratic reservations about Personal Good.}}

\textcolor{black}{In this section we will consider three objections to our argument that might motivate rejecting one of these principles rather than incompleteness, and explain why we find them unconvincing.}

\paragraph{Objection} An intuitively strange feature of the preceding arguments is that they use principles about the ranking of lotteries (Negative Dominance and Personal Good) to derive strong constraints on the rankings of lives and outcomes. This might seem to get things backwards: The betterness relations on lives and outcomes are more fundamental than the betterness relation on lotteries 
and so, one might claim, principles about the latter relation should be constrained by rather than constrain the former relations. 
Moreover, \emph{personal} betterness is more fundamental than \emph{impersonal} betterness, so it's especially strange for principles about the impersonal ranking of lotteries to constrain the personal ranking of lives. The comparative personal value of, say, the lives of an artist and a banker is a prior question to principles like Negative Dominance and Personal Good. If those lives are incomparable, that is a fact that ethical and decision-theoretic principles must bend to accommodate, not something that should be denied because it forces ethicists and decision theorists to make difficult choices.

\textcolor{black}{\paragraph{Response} This argument rests on an implicit premise we reject: that the relevant kind of fundamentality---be it metaphysical or some other form of ethical fundamentality---restricts what sorts of propositions 
can be used as \emph{evidence} for which others. Our argument does not rely on the claim that, metaphysically speaking, facts about the betterness of lotteries \emph{ground} or \emph{explain} the betterness of lives and outcomes. It requires only that structural principles governing the 
ranking of lotteries can give us \emph{evidence} about structural principles governing the rankings of lives and outcomes. Elementary particles are metaphysically more fundamental than galaxies, but we can still revise our beliefs about elementary particles on the basis of evidence provided by the behavior of galaxies.}

\paragraph{Objection} The arguments of \S \S \ref{s:resultsForPersonalIncompleteness}--\ref{s:resultsForPersonalCompleteness} relied on the combination of Negative Dominance and Personal Good. Although these two principles can be jointly accepted 
in the finite context that has been assumed thus far, 
they turn out to be inconsistent in certain infinitary contexts---specifically, 
when comparing lotteries with infinitely many possible outcomes that involve infinite or unboundedly large populations. 

\begin{table}
\centering

\begin{subtable}{\linewidth}
\centering
\caption{\textsc{Lottery 1: Better for All, but Surely Bad}}
\vskip 1mm
\begin{tabular}{ccccccc>{\columncolor{lightgray}}c}
\toprule
& $s_1$ $\left( \frac{1}{3}\right)$ & $s_2$ $\left( \frac{2}{9}\right)$ & $s_3$ $\left( \frac{4}{27}\right)$ & $s_4$ $\left( \frac{8}{81}\right)$ & $s_5$ $\left( \frac{16}{243}\right)$ & $\cdots$ &  $\operatorname{Pr}(+|\exists)$ \\
\midrule
$P_{1}$ & $-1$  & $1$ & $1$ & $1$  & $1$ & $\cdots$ & \sfrac{2}{3}\\
$P_{2-3}$  &  --- & $-1$ & $1$ & $1$  & $1$ & $\cdots$ & \sfrac{2}{3} \\
$P_{4-7}$  &  --- & --- & $-1$ & $1$  & $1$ &  $\cdots$ & \sfrac{2}{3} \\
$P_{8-15}$  &  --- & --- & --- & $-1$  & $1$ &  $\cdots$ & \sfrac{2}{3} \\
$P_{16-31}$  &  --- & --- & --- &  --- & $-1$ & $\cdots$ & \sfrac{2}{3} \\
$\vdots$ 				 & $\vdots$ & $\vdots$ & $\vdots$ & $\vdots$ & $\vdots$ & $\ddots$& $\vdots$  \\
\rowcolor{lightgray}
Total						& $-1$ & $-1$ & $-1$ & $-1$ & $-1$ & $\cdots$ & \\

\bottomrule
\end{tabular}
\end{subtable}

\vskip 4mm

\begin{subtable}{\linewidth}
\centering
\caption{\textsc{Lottery 2: Worse for All, but Surely Good}}
\vskip 1mm
\begin{tabular}{ccccccc>{\columncolor{lightgray}}c}
\toprule
& $s_1$ $\left( \frac{1}{3}\right)$ & $s_2$ $\left( \frac{2}{9}\right)$ & $s_3$ $\left( \frac{4}{27}\right)$ & $s_4$ $\left( \frac{8}{81}\right)$ & $s_5$ $\left( \frac{16}{243}\right)$ & $\cdots$ &  $\operatorname{Pr}(+|\exists)$ \\
\midrule
$P_{1}$ & $1$  & $-1$ & $-1$ & $-1$  & $-1$ & $\cdots$ & \sfrac{1}{3} \\
$P_{2-3}$  &  --- & $1$ & $-1$ & $-1$  & $-1$ & $\cdots$ & \sfrac{1}{3} \\
$P_{4-7}$  &  --- & --- & $1$ & $-1$  & $-1$ & $\cdots$ & \sfrac{1}{3} \\
$P_{8-15}$  &  --- & --- & --- & $1$  & $-1$ & $\cdots$ & \sfrac{1}{3} \\
$P_{16-31}$  &  --- & --- & --- &  --- & $1$ & $\cdots$ & \sfrac{1}{3} \\
$\vdots$ 				 & $\vdots$ & $\vdots$ & $\vdots$ & $\vdots$ & $\vdots$ & $\ddots$& $\vdots$  \\
\rowcolor{lightgray}
Total						& $1$ & $1$ & $1$ & $1$ & $1$ & $\cdots$ & \\

\bottomrule
\end{tabular}
\end{subtable}

\caption{\hspace{1mm} An illustration of the conflict between Personal Good and Negative Dominance (and other dominance principles) in contexts involving infinitely many possible outcomes and unbounded/infinite populations. Each individual has the same probability of existing in both lotteries (with ``---'' representing nonexistence), and two possible welfare levels conditional on existence. The rightmost column gives each individual's total probability of positive welfare, conditional on existence. The bottom row gives the total welfare realized in each state.}
\label{tab:personalGoodVsDominance}
\end{table}

Table \ref{tab:personalGoodVsDominance} illustrates this inconsistency. The table describes a pair of lotteries in which, for every possible individual, the personal lottery given by Lottery 1 stochastically dominates the personal lottery given by Lottery 2; and yet Lottery 2 is certain to result in a positive total of personal value, while Lottery 1 is certain to result in a negative total of personal value.\footnote{This example builds on one in \cite{kowalczykFCsaving}. Related examples of conflicts between dominance principles and \emph{ex ante} Pareto-like principles in infinitary contexts appear in \cite{goodsell2021st}, \cite{wilkinson2023infinite}, \cite{nebelms}, and \cite{hongFCparadoxes}.}  
Personal Good implies that Lottery 1 is better than Lottery 2. But it is plausible that every outcome of Lottery 2 is better than every outcome of Lottery 1, in which case Negative Dominance implies that  Lottery 1 cannot be better. This latter claim 
would follow straightforwardly given a total utilitarian theory of impersonal value, but it follows from many other views as well and is in any case intuitively plausible.\footnote{ 
Totalism provides a simple illustration of the argument in the main text, but the axiological assumptions required to generate this sort of conflict between Negative Dominance and Personal Good are very weak. It would suffice, for instance, to assume a principle we might call ``Inverse Lives'': There exists a pair of lives $l^+$ and $l^-$ such that, among outcomes in which everyone who exists has one of those two lives, any outcome where most people have $l^+$ is better than any outcome where most people have $l^-$. This assumption is endorsed, for instance, by standard versions of total utilitarianism, critical-level utilitarianism, prioritarianism, average utilitarianism, variable-value utilitarianism, and many forms of egalitarianism.

We have presented the example as involving only finite (though unboundedly large) populations. 
We could switch to an infinite-population example by replacing every instance of nonexistence in Table \ref{tab:personalGoodVsDominance} with a neutral life, with welfare $0$. This would allow us to weaken Personal Good by restricting it to fixed-population contexts (where the same individuals exist in every outcome of both lotteries). Then the principle could simply assert that if one lottery gives every individual a stochastically dominant personal lottery, then it's impersonally better.} While it is not obvious \emph{which} principle should be given up in this case, it shows that Negative Dominance and Personal Good cannot \emph{both} be unrestrictedly true. \textcolor{black}{And if a \textit{prima facie} plausible principle must be abandoned in one context, that \textcolor{black}{significantly} undercuts its appeal in other contexts as well.}

\textcolor{black}{\paragraph{Response} 
We reject the premise that needing to give up a principle in infinite contexts must significantly undercut its appeal in finite contexts.\footnote{\textcolor{black}{For discussion of this issue, see for instance \cite{russell2023fanaticism} and \cite{kowalczykFCsaving}.}} There are many principles that hold in finite but not infinite domains---for instance, famously, the principle that no set can be put into one-to-one correspondence with any of its proper subsets. As a closer analogy, although Pareto and 
the standard Anonymity principle 
(which says that permuting a welfare distribution in an outcome results in an equally good outcome) are incompatible for infinite populations \citep{vanliedekerke1995should}, most ethicists and social welfare theorists accept both of these principles for finite populations, and use them happily as premises in that context. So although either Negative Dominance or Personal Good must be given up in certain infinite contexts, this need not seriously weaken the case for completeness in finite ones. In particular, Negative Dominance and Personal Good are perfectly consistent when restricted to lotteries with finitely many possible outcomes, even when some of those outcomes involve infinite populations.}\footnote{\textcolor{black}{Here's a theory of value that satisfies both principles for infinite-population, finite-support lotteries: Let the set of possible individuals be indexed by the natural numbers, and assume a free ultrafilter $U$ on the naturals. Suppose further that the personal betterness ordering of lives is complete, and that the value of lives (including nonexistence) can be represented by real numbers. Let $\operatorname{EV}(P_i | L_j)$ be the expected personal value for individual $P_i$ under lottery $L_j$. Then the following betterness relation on lotteries will satisfy both Negative Dominance and Personal Good: $L_1$ is at least as impersonally good as $L_2$ if and only if the set of $n$ such that $\sum_{i = 0}^n \operatorname{EV}(P_i | L_1) \geq \sum_{i = 0}^n \operatorname{EV}(P_i | L_2)$ belongs to the ultrafilter $U$. This view satisfies Personal Good unrestrictedly \textcolor{black}{and satisfies Negative Dominance for finite-support lotteries}. 
}

\textcolor{black}{An implication of this consistency is that our result in \S \ref{s:resultsForPersonalCompleteness} ruling out strong incomparability between distributions \emph{could} be extended to infinite distributions---ruling out incomparability between, for instance, the distributions $(1,0,1,0,\ldots)$ and $(0,1,0,1,\ldots)$. 
But we have some reservations about the application of Personal Good to infinite-population, finite-support lotteries, as we explain 
in the next note.}}
\textcolor{black}{The proofs of our results in \S \S \ref{s:resultsForPersonalIncompleteness}--\ref{s:lotteries} used only lotteries with finite support, and so don't require Negative Dominance or Personal Good to hold outside that context.}\footnote{\textcolor{black}{\label{infinitenote} Then why restrict ourselves to finite populations, as well as finitely supported lotteries? With respect to our result in \S \ref{s:lotteries} (concerning the ranking of lotteries), the main reason is that the principle of Substitution is dubious in infinite-population contexts (see n.\ \ref{fn:infiniteSubstitution}). With respect to the result in \S \ref{s:resultsForPersonalCompleteness} (concerning the ranking of outcomes), the story is slightly more subtle.}

\textcolor{black}{First, \textcolor{black}{we currently incline to the view that} the principle to give up in Table \ref{tab:personalGoodVsDominance} is Personal Good. Notably, in this example, Personal Good conflicts with just about any imaginable dominance principle, including Stochastic Dominance, Statewise Dominance, and even Superdominance---the principle that, if each possible outcome of $L_1$ is better than every possible outcome of $L_2$, then $L_1$ is better than $L_2$. Restricting Negative Dominance therefore wouldn't do us any good here unless we were also willing to restrict those other dominance principles. 
And those principles are, \textcolor{black}{intuitively}, 
absolutely essential to the idea that the final value of outcomes should be our guide to assessing lotteries. On the other hand, we can tell a plausible story about why Personal Good should fail here: A stochastic improvement to one individual's personal lottery 
is a \textit{pro tanto} improvement to the lottery as a whole. But in certain cases, infinitely many \textit{pro tanto} improvements can add up to an overall worsening---this phenomenon crops up in various well-known infinitary decision problems, like Rouble Trouble and Satan's Apple \citep{arntzenius2004bayesianism}. Table \ref{tab:personalGoodVsDominance} is, we think, another case of this kind.}

\textcolor{black}{If that's the right story to tell, then it also casts some doubt on Personal Good in the context of 
finite-support lotteries over infinite populations---more specifically, for pairs of lotteries such that transforming one into the other would require permutations or other changes to infinitely many personal lotteries. Even if each of these changes in isolation is neutral or positive, their overall effect \emph{might} be non-positive. 
Thus, although Negative Dominance and Personal Good are consistent for finite-support lotteries on infinite populations, we think our arguments for completeness are less compelling in that context than in strictly finite contexts.}}

\vspace*{-\topskip}\textcolor{black}{\paragraph{Objection} Negative Dominance should be rejected, because its appeal is fully captured by 
the following, weaker principle:}

\begin{description}
\item[Supervaluational Negative Dominance] Lottery $L_1$ is better than lottery $L_2$ only if, on every admissible completion of the personal and impersonal betterness relations, some possible outcome $o_1$ of $L_1$ is \textcolor{black}{better than} some possible outcome $o_2$ of $L_2$.
\end{description}

An ``admissible completion'' is one that satisfies some set of fundamental constraints on the ranking of outcomes. 
Suppose these constraints include Transitivity and 
Pareto. Then, in the case in Table \ref{tab:personalIncomp}, every admissible completion of the betterness rankings of lives and outcomes must rank some outcome of Lottery 2 (either $(a^+,b)$ or $(b,a)$) above some outcome of Lottery 1 (either $(a,a)$ or $(b,b)$).\footnote{Suppose that, on a given completion, $a$ is at least as personally good as $b$. Then, by Transitivity, $a^+$ is better than $b$, so by Pareto, $(a^+,b)$ is better than $(b,b)$. The other possibility is that $b$ is better than $a$, in which case $(b,a)$ is better than $(a,a)$ by Pareto.}
Thus, unlike Negative Dominance, Supervaluational Negative Dominance does not rule out that Lottery 2 is better than Lottery 1. More generally, we conjecture that Supervaluational Negative Dominance is consistent, in finite contexts, with the various sets of assumptions that were shown to conflict with Negative Dominance in \S \S \ref{s:resultsForPersonalIncompleteness}--\ref{s:resultsForPersonalCompleteness}, in which case a retreat to Supervaluational Negative Dominance would block our arguments for completeness. 
\textcolor{black}{And one might think that Supervaluational Negative Dominance captures everything that is appealing about Negative Dominance: If every (admissible) way of filling in the gaps in the betterness relations makes some outcome of $L_1$ better than some outcome of $L_2$, then ``one way or another'' there will be some outcome-based rationale for the superiority of $L_1$, even if the actual, incomplete betterness relation on outcomes does not single out particular outcomes to furnish that rationale.}

\textcolor{black}{\paragraph{Response} There are two reasons why we regard Supervaluational Negative Dominance as an inadequate substitute for Negative Dominance. First, it does not capture the motivation for Negative Dominance that we gave in \S \ref{s:negativeDominance}. The idea central to our defense of Negative Dominance was that what ultimately matters in evaluating lotteries is how things actually turn out, and how things actually turn out will be a single outcome, not a disjunction or mixture of outcomes. Within our preference-based analysis of the value of lotteries, retreating to Supervaluational Negative Dominance would mean that you can get a requiring reason to prefer $L_1$ over $L_2$ from the fact that ``If you were to fill in the gaps in the betterness relation, then \emph{some} outcome of $L_1$ (either this outcome, or that outcome, or…) would be better than \emph{some} outcome of $L_2$.’’ It is exactly this appeal to disjunctions of outcomes that, we think, a focus on actual value precludes.}

\textcolor{black}{Second, Supervaluational Negative Dominance gives a normative significance to \emph{completions} that seems unwarranted absent a commitment to \emph{completeness}. To illustrate what we mean, let's consider a potential rationale for Supervaluational Negative Dominance from the recent literature.   \citeauthor{rabinowicz2008value} (\citeyear{rabinowicz2008value,rabinowicz2021incommensurability}) gives a developed theory of value relations, on which value relations in general (on outcomes as well as lotteries) are understood in terms of ideal preference relations. In his \citeyear{rabinowicz2021incommensurability}, he endorses the following normative constraint on incomplete preferences: 
\begin{description}
\item[Well-Roundedness] If there is a permissible completion of $P$ [an agent's preference relation] with respect to $G$ [a set of items in the domain of that preference relation], then $P$ is the intersection of some non-empty set of permissible completions of $P$ with respect to $G$. \citep[219]{rabinowicz2021incommensurability}
\end{description}
}

\textcolor{black}{This principle appears to support Supervaluational Negative Dominance but not Negative Dominance. To see \textcolor{black}{that it supports Supervaluational Negative Dominance}, suppose there is an admissible completion of the personal and impersonal betterness relations on which no outcome of $L_1$ is better than any outcome of $L_2$. Assuming that an admissible completion of the betterness relation corresponds to a permissible completion of the agent's preferences, and that all permissible completions of the agent's preferences will satisfy at least very minimal (``positive'') dominance principles, there will be a permissible completion of the agent's preferences on which $L_2$ is strictly preferred to $L_1$, so a preference for $L_1$ over $L_2$ will not belong to the intersection of the set of permissible completions.} 

\textcolor{black}{On the other hand, to see that Well-Roundedness does not support Negative Dominance, consider our central example from \S\S \ref{s:intro}/\ref{s:resultsForPersonalIncompleteness}. If the permissible completions of an agent's preference relation must all satisfy Personal Good (or even a weakening of Personal Good that applies only to personal lotteries over totally ordered lives), 
then those completions will all include a preference for the sweetened Lottery 2 over the unsweetened Lottery 1. \textcolor{black}{Thus that preference will belong to the intersection of any set of permissible completions.} So Rabinowicz's principle \textcolor{black}{at least leaves open} 
that one is required to prefer Lottery 2. And this would contradict Negative Dominance given Rabinowicz's (and our) preference-based understanding of the betterness relation on lotteries, \textcolor{black}{together with the auxiliary assumptions of Incomparability Transmission and Small Improvements}.}

\textcolor{black}{But the principle of Well-Roundedness, we think, gives an unwarranted normative significance to \emph{completions}. 
In the cases where it applies (where its antecedent is satisfied), Rabinowicz's principle can be paraphrased as follows: ``Suppose it's required that, \emph{if} the agent's preferences are complete, \emph{then} they include preference $x$. Then preference $x$ is required \textit{simpliciter}.'' Why accept this principle, if completeness itself is not required? There are many optional features of a preference relation such that any permissible preferences with that feature would include some particular preference. But if the feature itself is optional, then that particular preference may be too.}

\textcolor{black}{In short, we find Rabinowicz's principle of Well-Roundedness less compelling than the premises of our argument for Negative Dominance. More generally, Supervaluational Negative Dominance is too weak to capture the core motivations for Negative Dominance.}

\section{Conclusion}
\label{s:conclusion}

Against weak background assumptions, Negative Dominance and Personal Good together rule out 
\textcolor{black}{incompleteness} in the personal value of lives and the impersonal value of finite outcomes and lotteries. We see this as a compelling argument against those forms of \textcolor{black}{incompleteness}. Of course, as with any deductive argument, some might find the considerations opposing the conclusion so strong that they find it better to reject an initially plausible premise than to accept the conclusion. 
But even those who respond in this way to our arguments will have learned something surprising. 
If they reject Negative Dominance, they will 
have to give up the very natural idea that the comparative value of lotteries must be explained at least in part by the comparative value of their individual outcomes. And they will have to reject the stronger principle of Statewise Negative Dominance, 
which has figured centrally in the approaches to risk advocated by many fans of incompleteness
\citep{schoenfield2014decision,bales2014decision,doodyms}. On the other hand, if they reject Personal Good (even when restricted to situations where equality and diversity are not at stake), they will be giving up even the most minimal connection between the overall value of lotteries and their value for particular individuals. 

\textcolor{black}{More generally, our arguments show that proponents of standard forms of evaluative incompleteness must choose between the outcome-based perspective on risky choices represented by Negative Dominance, and the ``individualist'' perspective represented by Personal Good. These perspectives conflict in infinite contexts even without incompleteness. But one might have hoped to confine such conflicts to the paradox-ridden domain of infinity. Incompleteness forces us to abandon one of these perspectives even for strictly finite choices. And, we have argued, this is a substantial  cost.}\footnote{An earlier version of this paper was circulated under the title ``Share the Sugar''. For helpful discussion and feedback, we are grateful to Ralf Bader, Pietro Cibinel, Will Combs, Johan Gustafsson, Martin Peterson, Ralph Wedgwood, two anonymous referees, and audiences at the Global Priorities Institute, the 2024 meeting of the Philosophy, Politics, and Economics Society, Texas A\&M, the University of Toronto, and the University of Southern California.}

\begin{small}
\begin{singlespacing}

\bibliographystyle{chicago-custom}
\bibliography{references}
\end{singlespacing}

\end{small}

\end{document}